\documentclass[12pt]{article}

\usepackage[margin=2.5cm]{geometry}

\usepackage{amsmath,amsthm,amsfonts,amscd,amssymb,bbm,mathrsfs,enumerate,url}

\usepackage{graphicx,tikz}
\usepackage[pdfborder={0 0 0}]{hyperref}
\usetikzlibrary{matrix,arrows}

\def\RR{{\mathbb R}}
\def\CC{{\mathbb C}}

\def\ZZ{{\mathbb Z}}

\def\A{{\mathcal A}}
\def\B{{\mathcal B}}
\def\C{{\mathcal C}}

\def\H{{\mathcal H}}
\def\I{{\mathcal I}}
\def\K{{\mathcal K}}
\def\M{{\mathcal M}}
\def\N{{\mathcal N}}

\def\P{{\mathcal P}}
\def\R{{\mathcal R}}

\def\a{\alpha}

\def\k{\kappa}

\def\l{\lambda}

\def\s{\sigma}

\def\Ad{{\hbox{\rm Ad\,}}}

\def\dim{{\hbox{dim}\,}}

\def\1{{\mathbbm 1}}

\def\diff{{\rm Diff}}

\def\diffs1{\diff(S^1)}
\def\mob{{\rm M\ddot{o}b}}

\def\supp{{\rm supp\,}}

\def\psl2r{{\rm PSL}(2,\RR)}
\def\sl2r{{\rm SL}(2,\RR)}
\def\su11{{\rm SU}(1,1)}
\def\2dmob{{\overline{\psl2r}\times\overline{\psl2r}}}
\def\<{\langle}
\def\>{\rangle}

\def\poincare{{\P^\uparrow_+}}

\def\botimes{\mathbin{\bar{\otimes}}}

\DeclareMathOperator{\Tr}{Tr}

\newcommand{\tout}{\mathrm{out}}
\newcommand{\tin}{\mathrm{in}}

\newtheorem{theorem}{Theorem}[section]

\newtheorem{proposition}[theorem]{Proposition}
\newtheorem{lemma}[theorem]{Lemma}
\theoremstyle{remark}
\newtheorem{remark}[theorem]{Remark}

\title{Free products in AQFT}
\date{} 
\author{
{\bf Roberto Longo\footnote{Supported in part by the ERC Advanced Grant 669240 QUEST ``Quantum Algebraic Structures and Models'',
PRIN-MIUR, GNAMPA-INdAM and Alexander von Humboldt Foundation.},
Yoh Tanimoto\footnote{Supported by Programma Rita Levi Montalcini of the Italian Ministry of Education, University and Research.}} \\
   Dipartimento di Matematica, Universit\`a di Roma Tor Vergata\\
   Via della Ricerca Scientifica 1, I-00133 Roma, Italy\\
   email: {\tt longo@mat.uniroma2.it, hoyt@mat.uniroma2.it} 
      \vspace{0.5cm}\\
 {\bf Yoshimichi Ueda\footnote{Supported in part by the Grant-in-Aid for Challenging Exploratory Research 16K13762.}} \\
  Graduate School of Mathematics, Kyushu University \\
  Nishi-ku Motooka 744, Fukuoka, 810-8560, Japan \\
   email: {\tt ueda@math.kyushu-u.ac.jp}
}

\begin{document}
\maketitle

\begin{abstract}
We apply the free product construction to various local algebras in algebraic quantum field theory.

If we take the free product of infinitely many identical half-sided modular inclusions with ergodic canonical endomorphism,
we obtain a half-sided modular inclusion with ergodic canonical endomorphism and trivial relative commutant. 
On the other hand, if we take M\"obius covariant nets with trace class property,
we are able to construct an inclusion of free product von Neumann algebras with large relative commutant,
by considering either a finite family of identical inclusions or an infinite family of inequivalent inclusions.
In two dimensional spacetime, we construct Borchers triples with trivial relative commutant
by taking free products of infinitely many, identical Borchers triples.
Free products of finitely many Borchers triples are possibly associated with Haag-Kastler net having S-matrix which is
nontrivial and non asymptotically complete, yet the nontriviality of double cone algebras remains open.
\end{abstract}

\section{Introduction}\label{introduction}
The whole model of quantum field theory (QFT), in the framework of algebraic QFT,
can be reconstructed from just a combination of few local algebras
and the vacuum state.
This was a striking discovery by Wiesbrock and others \cite{Wiesbrock93-1, GLW98, Wiesbrock98, KW01, AZ05}.
In particular, a half-sided modular inclusion (HSMI) of two half-line algebras coming from
the (strongly additive) chiral component of a two-dimensional conformal field theory completely remembers
the structure of QFT on the lightray. It opened a hope to construct new models of QFT which is in general
very difficult, and since then, HSMIs have been used to construct possibly new chiral components
\cite{Longo01, CLTW12-2} as well as to determine the structure of massless two-dimensional QFT \cite{Tanimoto12-1, Tanimoto12-2}.
Actually, a HSMI is not equivalent to a chiral component (M\"obius covariant net), and in order to carry out the reconstruction,
one needs an additional condition: the relative commutant of the algebras must be also large.
The contrary case does happen, for example when one considers wedge inclusions of (two-dimensional) fermionic QFT \cite{BLM11}
or bosonic interacting QFT \cite{BT15, Tanimoto14-1}. Yet, even if the relative commutant is not large enough
(the vacuum vector is not cyclic for it), if it is just nontrivial, one can still construct a M\"obius covariant net
on a smaller subspace \cite{BLM11}. Constructing HSMIs and determining the relative commutant are not easy tasks,
and in all the known cases, the relative commutant have been nontrivial.
In the present work, we provide examples of HSMIs with trivial relative commutant, by exploiting the techniques
of free product von Neumann algebras.

The study of free products of (possibly type III) von Neumann algebras dates back to the birth of free probability theory initiated by Voiculescu (see e.g.\! \cite{VDN92}), and their fundamental properties like central decomposition and type classification in the not necessarily type II$_1$ setting were recently established in full generality (see \cite{Ueda11-1}).
From a family $\{(\M_\k, \Omega_\k)\}_{\k \in K}$ of von Neumann algebras equipped with vector states,
one constructs $(\M, \Omega)$ where $\M$ is generated by the faithful images of $\M_\k$ and
these images are freely independent, or roughly speaking, they do not commute in the highest degree but still have some ``relations'' provided by $\Omega$.
One can actually start with a family of inclusions $(\N_\k \subset \M_\k, \Omega_\k)$
and obtains an inclusion $\N \subset \M$ of free product von Neumann algebras.
We will prove that if $(\N_\k \subset \M_\k, \Omega_\k)$ is an {\it infinite} family of {\it identical} HSMIs,
the resulting inclusion $\N \subset \M$ must have trivial relative commutant.

Before this work, no serious investigation of inclusions of free product von Neumann algebras has been made without assuming the existence of faithful normal conditional expectations.
We provide an illustrative example whereby the relative commutant is nontrivial by exploiting the techniques of AQFT. This is quite an interesting phenomenon, because it seems, at first sight, contrary to the highest degree of non-commutativity of the free product construction. 
If $\{(\{\M_\k(I)\}_{I\subset S^1}, \Omega_\k)\}_{\k\in K}$ is a family of (local) M\"obius covariant nets with vacuum vectors,
we can construct a family of von Neumann algebras satisfying the axioms of M\"obius covariant net except locality.
Then, by taking $K$ to be finite or considering an appropriate family of local algebras with infinite $K$,
we obtain inclusions $\N \subset \M$ of free product von Neumann algebras where $\N'\cap \M$ is a type III$_1$ factor.
The proof utilizes the $L^2$-nuclearity condition \cite{BDL90}, a physical
condition saying that the state space with small energy is not too large. 
Although the $L^2$-nuclearity condition indeed guarantees the nontriviality of the relative commutant, we do not find any explicit, nontrivial elements in it.

We also apply a similar construction to two-dimensional quantum field theory. A two-dimensional Haag-Kastler net can be
reconstructed from a single von Neumann algebra, a vacuum state and a representation of translations (a Borchers triple).
One can again consider the free product of Borchers triples.
Here, if the index set is infinite and all the given triples are identical, we obtain a Borchers triple with trivial relative commutant.
On the other hand, if the index set is finite, we are not able to determine the relative commutant.
Yet, one can define and compute the two-particle scattering S-matrix, and it turns out to be nontrivial
and not asymptotically complete. This case is of particular interest, because it might give an example of
Haag-Kastler net with minimal size, only above which the algebras of local observables are nontrivial.

This paper is organized as follows.
In Section \ref{preliminaries} we review the fundamental concepts of algebraic QFT, half-sided modular inclusions
and free products of von Neumann algebras. In Section \ref{hsmi}, we construct a HSMI with trivial relative commutant by
taking the free product of infinitely many copies of a standard HSMI. Section \ref{moebius} introduces
the concept of free product M\"obius covariant nets, and this concept enables us to give an example of inclusion of free product
von Neumann algebras with large relative commutant.
Finally, in Section \ref{borchers} we consider free products of two-dimensional Borchers triples.
When the index set $K$ is infinite and the family consists of identical ones, the relative commutant is shown to be trivial.
When $K$ is finite, it remains open, meaning that it might give a two-dimensional Haag-Kastler net
which should not have the wedge split property for small distance.

\section{Preliminaries}\label{preliminaries}
\subsection{Algebraic QFT}\label{AQFT}

\subsubsection{One-dimensional chiral components}\label{chiral}

Let $S^1$ denote the unit circle and let $\mathcal{I}$ be the collection of proper intervals
(i.e., open, connected, non-empty, non-dense subsets) $I \subset S^1$. For any $I \in \mathcal{I}$,
we denote by $I'$ the interior of the complement of $I$ in $S^1$.
The M\"obius group $\mob := \mathrm{SL}(2,\mathbb{R})/\mathbb{Z}_2 \cong \mathrm{PSU}(1,1)$
acts on $S^1$ by diffeomorphisms in a natural fashion. Let $R : \mathbb{R}/2\pi\ZZ \to \mob$ denote the rotation subgroup.      
Let $\mathcal{A} = (\mathcal{A}(I))_{I \in \mathcal{I}}$ be a
{\bf M\"obius covariant net} (or a M\"obius covariant precosheaf) on $S^1$ in the sense of \cite[Section 2]{DLR01}.
Namely, all the $\mathcal{A}(I)$, $I \in \mathcal{I}$, are von Neumann algebras on a fixed Hilbert space
$\mathcal{H}$ with the following properties: 
\begin{itemize}
\item[1.] Isotony: $I_1 \subset I_2$ $\Rightarrow$ $\mathcal{A}(I_1) \subset \mathcal{A}(I_2)$. 
\item[2.] M\"obius covariance: There exists a unitary representation $U: G \curvearrowright \mathcal{H}$ such that
$U(g)\mathcal{A}(I)U(g)^* = \mathcal{A}(gI)$ for every $g \in G$ and $I \in \mathcal{I}$. 
\item[3.] Positivity of the energy: The generator $L_0$ of the one parameter unitary group $\theta \mapsto U(R(\theta)) = e^{i\theta L_0}$,
called the {\bf conformal Hamiltonian}, is positive. 
\item[4.] Existence of the vacuum: There exists a unique (up to a scalar) unit vector $\Omega \in \mathcal{H}$,
called the {\bf vacuum vector}, such that $U(g)\Omega = \Omega$ for all $g \in \mob$
and that $\Omega$ is cyclic for $\bigvee_{I \in \mathcal{I}} \mathcal{A}(I)$ and separating
for $\bigwedge_{I \in \mathcal{I}} \mathcal{A}(I)$.   
\end{itemize}
A M\"obius covariant net $\mathcal{A}$ is said to be {\bf local} if the following property holds:  
\begin{itemize}
\item[5.] Locality: $I_1 \cap I_2 = \emptyset$ $\Rightarrow$ $\mathcal{A}(I_1) \subseteq \mathcal{A}(I_2)'$. 
\end{itemize}
Algebras $\A(I)$ are of type III$_1$, unless the Hilbert space is one-dimensional (we call this case {\bf trivial})
\cite[Proposition 2.4(ii)]{DLR01}. 
{\bf The Bisognano-Wichmann property} 
$\Delta_I^{it} := \Delta_{\A(I)}^{it} = U(\Lambda_I(-2\pi t))$ holds,
where $\Delta_{\A(I)}$ is the modular operator of $\A(I)$ with respect to $\Omega$ and
$\Lambda_I$ is the dilation associated with $I$, see \cite[Proposition 2.2(ii)]{DLR01}. 
In the following, we always stress whether each statement holds with or without locality,
so that no confusion arises.

In a nice situation, the whole M\"obius covariant net can be reconstructed from a pair of von Neumann algebras.
Let us consider another additional property.
\begin{itemize}
\item[6.] Strong additivity: For $I \in \I$ and $I_1, I_2$ the two intervals obtained by removing one point from $I$,
it holds that $\A(I) = \mathcal{A}(I_1) \vee \mathcal{A}(I_2)$. 
\end{itemize}

Let $\H$ be a Hilbert space, $\M$ a von Neumann algebra on $\H$, $T$ a unitary representation of $\RR$ on $\H$ with positive spectrum and
$\Omega$ a vector in $\H$. We say $(\M, T, \Omega)$ is a {\bf one-dimensional Borchers triple} on $\H$ if
\begin{itemize}
 \item $\Omega$ is the unique (up to a scalar) invariant vector under $T(a)$,
 \item $\Ad T(a) (\M) \subset \M$ for $a\in \RR_+$,
 \item $\Omega$ is cyclic and separating for $\M$.
\end{itemize}
This notion is related with the following:
a {\bf half-sided modular inclusion} (HSMI) \footnote{The definition here is called $+$HSMI in \cite{Wiesbrock98}.
The other case, the inclusion for $t \le 0$, is called $-$HSMI and can be treated in a completely parallel way.}
is an inclusion $\N \subset \M$ of von Neumann algebras and
a vector $\Omega$ which is cyclic and separating for both of $\N, \M$ such that
$\s_t^\M(\N) \subset \N$ for $t \ge 0$, where $\s_t^\M$ is the modular automorphism group of $\M$ with respect to $\Omega$.

If $(\M, T, \Omega)$ is a one-dimensional Borchers triple, then $\N = \Ad T(1)(\M) \subset \M$ is
a half-sided modular inclusion by the Borchers theorem \cite{Borchers92, Florig98}.
The uniqueness of $\Omega$ implies the ergodicity $\bigcap_{t \ge 0}\s_t^\M(\N) = \CC\1$.
Conversely, if $(\N \subset \M, \Omega)$ is a half-sided modular inclusion, then
one can construct a representation $T(a)$ of $\RR$ with positive spectrum such that $\N = \Ad T(1)(\M)$
and $T(2) = J_{\N}J_{\M}$, where $J_{\M}$ and $J_{\M}$ are the modular conjugations of
$\M$ and $\N$ with respect to $\Omega$, respectively
\cite{Wiesbrock93-1, AZ05}. Furthermore, $T(a)$ together with the modular group $\Delta_{\mathcal{M}}^{it}$ generates
a representation of the translation-dilation group.
We remark that $\s_t^\M$ is ergodic on $\N$ if and only if the canonical endomorphism
$\Ad \Gamma$ is ergodic on $\N$, where $\Gamma = J_{\N}J_{\M} = T(2)$, and in this case, the uniqueness of $\Omega$ follows,
giving rise to a one-dimensional Borchers triple.

We say that a HSMI $(\N\subset \M, \Omega)$ is {\bf standard}
if $\Omega$ is cyclic and separating for $\N'\cap \M$.
We also say that a one-dimensional Borchers triple is standard if so is the corresponding half-sided modular inclusion.
Let us identify $S^1$ and the one-point compactification of $\RR$ by the stereographic projection.
There is a one-to-one correspondence between standard half-sided modular inclusions $(\N\subset \M, \Omega)$
and strongly additive local M\"obius covariant nets $\A$ by
$\A((0,\infty)) = \M, \A((1,\infty)) = \N$ \cite[Corollary 1.7]{GLW98}.

Actually, even if $(\N\subset \M, \Omega)$ is not standard, one can still construct a local M\"obius covariant net
on the subspace $\overline{(\N'\cap \M)\Omega}$: indeed, one should just define $\A((0,1)) := \N'\cap \M$
and then $\A((0,\infty))$ by dilation covariance and closure. In this way, one obtains a standard one-dimensional
Borchers triple, hence a strongly additive local M\"obius covariant net (c.f.\! \cite[Lemma 5.1, Theorem 5.2]{BLM11}).

Therefore, it is an important problem to determine whether the relative commutant $\N'\cap \M$ of a HSMI is nontrivial.
Moreover, in all examples constructed so far, the relative commutant was either nontrivial or difficult to identify.
One of the main results in this work is to provide examples of half-sided modular inclusion with
trivial relative commutant.

\subsubsection{Two-dimensional nets}\label{two-dim}
In $(1+1)$-dimensional Minkowski space, a (Poincar\'e-covariant) {\bf Haag-Kastler net} is a family $\{\A(O)\}_{O \subset \RR^2}$
of von Neumann algebras on a fixed Hilbert space $\H$ parametrized by bounded open regions in $\RR^2$ satisfying the following properties:
\begin{enumerate}
\item Isotony: If $O_1 \subset O_2$, then $\mathcal{A}(O_1) \subset \mathcal{A}(O_2)$. 
\item Locality: If $O_1$ and $O_2$ are spacelike separated, then $\mathcal{A}(O_1)$ and $\mathcal{A}(O_2)$ commute. 
\item Poincar\'e covariance: There is a unitary representation $U$ of the proper orthochronous Poincar\'e group $\poincare$
on $\H$ such that $U(g)\mathcal{A}(O)U(g)^* = \mathcal{A}(gO)$ for every $g \in \poincare$ and open region $O \subset \RR^2$. 
\item Positivity of the energy: The restriction of $U$ to the translation subgroup $\RR^2$ has the joint spectrum included
in $\overline V_+ = \{(t_0,t_1) \in \RR^2: t_0 \ge |t_1|\}$. 
\item Existence of the vacuum: There exists a unique (up to scalar) unit vector $\Omega \in \mathcal{H}$,
called the {\bf vacuum vector}, such that $U(g)\Omega = \Omega$ for all $g \in \poincare$ and that
$\Omega$ is cyclic for $\mathcal{A}(O)$ where $O$ is sufficiently large.
\end{enumerate}
\begin{remark}
 Usually one assumes weak additivity and cyclicity of the vacuum for the global algebra $\bigcup_{O\subset \RR^2}\A(O)$,
 and proves the Reeh-Schlieder property $\overline{\A(O)\Omega} = \H$, where $O$ is open.
 Yet, the weak additivity is not necessarily a physical requirement, neither is it known whether it follows from the Wightman axioms.
 In addition, we are also interested in Haag-Kastler nets possibly with minimal size, namely,
 $\Omega$ is not cyclic for $\A(O)$ when $O$ is small.
 Therefore, we include the Reeh-Schlieder property for sufficiently large $O$ to the axioms.
\end{remark}

It is not an easy task to construct a Haag-Kastler net, especially the infinite family of von Neumann algebras $\{\A(O)\}$,
but in some situations it can be reduced to a single von Neumann algebra associated with the wedge $W_{\mathrm R} = \{(a_0,a_1): a_1 > |a_0|\}$.
By a {\bf two-dimensional Borchers triple} 
$({\mathcal M},T,\Omega)$ we mean a triple of a von Neumann algebra ${\mathcal M}$ on ${\mathcal H}$,
a unitary representation $T$ of ${\mathbb R}^2$
with joint spectrum in $\overline V_+$ and a vector $\Omega$ (called the vacuum vector) 
such that
\begin{itemize}
 \item $\Omega$ is the unique (up to scalar) invariant vector under $T(a)$.
 \item ${{\rm Ad\,}} T(a) {\mathcal M} \subset {\mathcal M}$ for $a \in W_{\mathrm R}$.
 \item $\Omega$ is cyclic and separating for $\M$.
\end{itemize}

It is immediate that if $({\mathcal A},U,\Omega)$ is a Poincar\'e covariant Haag-Kastler net,
then the triple $({\mathcal A}(W_{\mathrm R}), U|_{{\mathbb R}^2}, \Omega)$ is a Borchers triple.
Conversely, if $({\mathcal M},T,\Omega)$ is a Borchers triple, one can construct a net as follows.
Any double cone in two-dimensional spacetime is the intersection of two-wedges
$(W_{\mathrm R}+a)\cap (W_{\mathrm L}+b) =: D_{a,b}$, where $W_{\mathrm L}$ is the reflected (left-)wedge.
Then one constructs the net first for double cones by
${\mathcal A}(D_{a,b}) := {{\rm Ad\,}} T(a)({\mathcal M}) \cap {{\rm Ad\,}} T(b)({\mathcal M}')$, and then 
for a general region $O$ by ${\mathcal A}(O) := \left(\bigcup_{D_{a,b}\subset O} {\mathcal A}(D_{a,b})\right)''$.
It is easy to show that this family ${\mathcal A(O)}$ satisfies isotony and locality.
Furthermore, the representation $T$ can be extended by the Tomita-Takesaki theory
to a representation $U$ of the Poincar\'e group such that $({\mathcal A},U)$ is covariant and $\Omega$ is still its fixed vector.
The requirement that $\Omega$ should be cyclic for sufficiently large $O$ needs to be checked in examples.
See \cite[Section 2]{Lechner15}.

\subsection{The split property and nuclearity conditions}

We say that an inclusion $\N \subset \M$ of von Neumann algebras is {\bf split} if there is an intermediate
type I factor $\R$, i.e.\!, $\N \subset \R \subset \M$.
If $\N$ and $\M$ are type III algebras on a separable Hilbert space $\H$,
the split property has an immediate consequence on the relative commutant:
$\N'\cap \M$ is nontrivial. Indeed, by identifying the intermediate type I factor $\R$ with $\B(\K_1)\otimes\CC\1$,
this inclusion is unitarily equivalent to $\N \otimes \CC \subset \B(\K_1)\otimes \M$, where $\M$ and $\N$ are type
III factors, hence we obtain $\N'\cap \M \cong \N'\otimes \M$.

There are several sufficient conditions for the split property in terms of the modular theory.
Let $\Omega$ be a common cyclic and separating vector for $\N$ and $\M$
and $\Delta_\M$ be the modular operator for $\M$ with respect to $\Omega$.
We say that the inclusion $(\N \subset \M, \Omega)$ satisfies the {\bf modular nuclearity} condition
if the map
\[
 \N \ni x \longmapsto \Delta_\M^\frac14 x \Omega
\]
is nuclear. The modular nuclearity condition implies the split property for $\N \subset \M$
\cite[Propositions 1.1, 2.3]{BDL90}.
Furthermore, if $\Delta_\M^\frac14 \Delta_\N^{-\frac14}$ is nuclear, where $\Delta_\N$ is the modular operator for $\N$
with respect to $\Omega$, the inclusion $\N \subset \M$ is said to satisfy the {\bf $L^2$-nuclearity} condition.
The $L^2$-nuclearity implies the modular nuclearity \cite[Proposition 5.3]{BDL07},
hence also the split property.

\paragraph{One-dimensional case.}
Let us turn to conformal nets $(\A, U, \Omega)$.
The half-sided modular inclusion $\A((1,\infty)) \subset \A((0,\infty))$ is never split
(see the argument of \cite[P.292(b)]{Buchholz74}).
Yet, one can consider inclusions $\A(I)\subset \A(\tilde I)$ with $\overline I \subset \tilde I$.

Let us focus on the generator of rotations $L_0$. It has a discrete spectrum, as the $2\pi$-rotation is trivial.
If $\Tr e^{-sL_0} < \infty$ for some $s$, then the \textit{distal split property} holds:
namely, if $I \subset \tilde I$ with ``conformal distance'' $\ell(I,\tilde I) > s$,
$\A(I) \subset \A(\tilde I)$ is split \cite[Corollary 6.4]{BDL07}.

Actually, for most of the examples of \textit{local} M\"obius covariant nets,
a stronger property holds: $\dim \ker (L_0-n)$ grows asymptotically as $e^{\a n^\nu}$ where $0 < \nu < 1$.
In such a case, $e^{-s n}\cdot e^{\a n^\nu}$ tends rapidly to $0$, hence
$\Tr e^{-sL_0} < \infty$ for any $s > 0$.

\paragraph{Two-dimensional case.}

Let $(\M, T, \Omega)$ be a two-dimensional Borchers triple.
There is a sufficient condition for $\Omega$ to be cyclic for double cone algebras $\A(D_{0,a}) = \M\cap \Ad T(a)(\M')$ \cite{BDL90}:
let $\Delta_\M$ be the modular operator for $\M$
with respect to $\Omega$. For $a \in W_{\mathrm R}$, consider the map
\[
 \M \ni x \longmapsto \Delta_\M^\frac14 U(a)x\Omega.
\]
If this map from $\M$ to $\H$ is nuclear, we say that the Borchers triple $(\M, U, \Omega)$ satisfies
\textbf{the modular nuclearity} (for $a$), and then there is a type I factor $\R_a$ such that
$\Ad U(a)(\M) \subset \R_a \subset \M$ \cite[Proposition 2.4]{BDL90}, \cite[Proposition 5.2]{BDL07}.
If $\mathrm{Ad}U(a)(\M) \subset \M$ is split for a sufficiently large $a$, then $\Omega$ is cyclic for
$\A(D_{0,a})$ \cite[Section 2]{Lechner08} (this statement is only implicit there).

\subsection{Free product of von Neumann algebras}\label{freeproduct}

Here we recall the free product construction in von Neumann algebras.
Nowadays, free product von Neumann algebras are usually understood in quite an abstract fashion
based on free independence, but a more constructive way starting with the notion of free products
of Hilbert spaces is appropriate in the context of AQFT.
Hence we summarize an older approach to free product von Neumann algebras.

Let $K$ be a possibly infinite index set and $\{\mathcal{H}_\kappa\}_{\kappa \in K}$
a family of Hilbert spaces with distinguished unit vectors $\Omega_\kappa \in \mathcal{H}_\kappa$.
With a one-dimensional subspace $\CC\Omega$ and a unit vector $\Omega$,
the free product $(\mathcal{H},\Omega)$ of these $(\mathcal{H}_\kappa,\Omega_\kappa)$, $\kappa \in K$,
is defined by
\begin{equation}\label{Eq2.1}
\mathcal{H} = \mathbb{C}\Omega \oplus \bigoplus_{n\geq 1} \bigoplus_{\underset{1 \leq j \leq n-1}{\kappa_j \neq \kappa_{j+1}}} \mathcal{H}^\circ_{\kappa_1}\botimes\cdots\botimes\mathcal{H}^\circ_{\kappa_n} \quad \text{with} \quad \mathcal{H}^\circ_\kappa := \mathcal{H}_\kappa\ominus\mathbb{C}\Omega_\kappa.   
\end{equation}
Let us denote by $P_{\CC\Omega}$ and $P_{\H_\k^\circ}$
the orthogonal projections onto $\CC\Omega$ and $\H_\k^\circ$, respectively.
For any given family $\{T_\k\}_{\k\in K}$ of bounded operators on $\mathcal{H}_\kappa$ (linear or antilinear),
with norm $\|T_\k\|$ 
not greater than $1$,
such that $T_\kappa \H^\circ_\k \subset \H^\circ_\k$,
we use the symbol 
$$
\bigstar_{\kappa \in K} T_\kappa := T_\Omega \oplus \bigoplus_{n\geq1}\bigoplus_{\underset{1 \leq j \leq n-1}{\kappa_j \neq \kappa_{j+1}}} T_{\kappa_1}^\circ\otimes\cdots\otimes T_{\kappa_n}^\circ \in \B(\mathcal{H}) \quad \text{with} \quad 
T_\kappa^\circ := T_\kappa\!\upharpoonright_{\mathcal{H}_\kappa^\circ},
$$
where $T_\Omega = \1_\CC$  if $T_\k$'s are linear and $T_\Omega = J_\CC$ (the complex conjugation)
if $T_\k$'s are antilinear
(according to the direct sum decomposition \eqref{Eq2.1}) following \cite[\S1.8]{Voiculescu85}. When $T_\k \Omega_\k = \Omega_\k$ and $T_\k^* \Omega_\k = \Omega_\k$
(a stronger assumption than $T_\kappa \H^\circ_\k \subset \H^\circ_\k$), the resulting  
$\bigstar_{\kappa \in K} T_\kappa$ restores $T_\k$ as its restriction to $\mathcal{H}_\k = \mathbb{C}\Omega \oplus \mathcal{H}_k^\circ \subset \mathcal{H}$ and hence it is a common extension of the given $T_\k$. For each $ \kappa \in K$ we have two normal representations $\lambda_\kappa : \B(\mathcal{H}_\kappa) \curvearrowright \mathcal{H}$
and $\rho_\kappa : \B(\mathcal{H}_\kappa) \curvearrowright \mathcal{H}$ acting from the left and the right, respectively,
as in \cite[\S1.2]{Voiculescu85}, and they enjoy the following commutation relation \cite[\S1.3]{Voiculescu85}:
For any $\kappa, \kappa' \in K$ and $T \in \B(\mathcal{H}_\kappa)$ and $T' \in \B(\mathcal{H}_{\kappa'})$ we have 
\begin{align*}
[\lambda_\kappa(T), \rho_{\kappa'}(T')] = \delta_{\kappa,\kappa'} (P_{\mathbb{C}\Omega} + P_{\mathcal{H}_\kappa^\circ})\lambda_\kappa([T,T']) = \delta_{\kappa,\kappa'} \rho_\kappa([T,T'])(P_{\mathbb{C}\Omega} + P_{\mathcal{H}_\kappa^\circ}).
\end{align*}
Define the unitary involution $Z : \mathcal{H} \to \mathcal{H}$ by $Z\Omega = \Omega$ and 
$$
Z(\xi_1\otimes\cdots\xi_n) = \xi_n\otimes\cdots\otimes\xi_1 \quad \text{(in reverse order)}
$$ 
for any $\xi_1\otimes\cdots\otimes\xi_n \in \mathcal{H}^\circ_{\kappa_1}\botimes\cdots\botimes\mathcal{H}^\circ_{\kappa_n} \subset \mathcal{H}$. As remarked in \cite[\S1.7]{Voiculescu85}, 
\begin{equation}\label{Eq2.2}
Z\lambda_\kappa(T) = \rho_\kappa(T)Z
\end{equation} 
holds for every $T \in \B(\mathcal{H}_\kappa)$. We will keep this notation throughout this paper. Let $\mathcal{M}_\kappa \subseteq \B(\mathcal{H}_\kappa)$, $\kappa \in K$, be von Neumann algebras such that each $\Omega_\kappa$ is cyclic for $\mathcal{M}_\kappa$.
Then, it is known, see \cite[Proposition 1.9]{Voiculescu85} due to Voiculescu, that 
\begin{equation}\label{Eq2.3}
\Bigg(\bigvee_{\kappa\in K} \lambda_\kappa(\mathcal{M}_\kappa)\Bigg)' = \bigwedge_{\kappa \in K} \lambda_\kappa(\mathcal{M}_\kappa)'  = \bigvee_{\kappa \in K} \rho_\kappa(\mathcal{M}_\kappa') 
\end{equation}
on $\mathcal{H}$, where only the second equality is nontrivial. 
Here the symbols $\bigvee$ and $\bigwedge$ denote the operation of generation as von Neumann algebra and the intersection operation, respectively. Remark that $\Omega$ is cyclic for $\bigvee_{\kappa\in K} \lambda_\kappa(\mathcal{M}_\kappa)$ by construction. Moreover, by the above commutation relation, we see that $\Omega$ is separating for $\bigvee_{\kappa \in K} \lambda_\kappa(\mathcal{M}_\kappa)$,
too, when every $\Omega_\kappa$ is separating for $\mathcal{M}_\kappa$.
The von Neumann algebra $\mathcal{M} = \bigvee_{\kappa \in K} \lambda_\kappa(\mathcal{M}_\kappa)$ is a concrete and standard realization of
the so-called {\bf free product von Neumann algebra} of the $(\mathcal{M}_\kappa,\omega_\kappa)$, $\kappa \in K$,
where the vector state $\omega_\kappa$ constructed from $\Omega_\kappa$ for every $\kappa \in K$. The vector state $\omega$ constructed from $\Omega$ is called the {\bf free product state}.

Abstractly, the {\bf free product} $(\mathcal{M},\omega) = \bigstar_{\kappa \in K} (\mathcal{M}_\kappa,\omega_\kappa)$ 
can be formulated by the following four conditions: 
\begin{itemize}
\item There exist normal injective $*$-homomorphisms $\lambda_\kappa : \mathcal{M}_\kappa \to \mathcal{M}$, $\kappa \in K$, whose ranges generate $\mathcal{M}$ as von Neumann algebra. 
\item $\omega\circ\lambda_\kappa = \omega_\kappa$ for every $\kappa \in K$. 
\item The $\lambda_\kappa(\mathcal{M}_\kappa)$, $\kappa \in K$, are {\bf freely independent} in $(\mathcal{M},\omega)$, that is, 
\begin{align}\label{eq:freeindep}
\omega(\lambda_{\kappa_1}(x_1) \cdots \lambda_{\kappa_n}(x_n)) = 0
\end{align}
whenever $x_j \in \mathrm{Ker}(\omega_{\kappa_j})\cap\mathcal{M}_{\kappa_j}$ with $\kappa_j \neq \kappa_{j+1}$ ($1 \leq j \leq n-1$) and $n \geq 1$. 
\item The GNS representation of $\mathcal{M}$ associated with $\omega$ is faithful. 
\end{itemize}
It is fundamental that the modular automorphism $\sigma_\mathcal{M}^t$ associated with $\omega$ can be determined by
the modular condition (the KMS condition with inverse temperature $-1$) as 
\begin{equation}\label{Eq2.4} 
\sigma^t_\mathcal{M}\circ\lambda_\kappa = \lambda_\kappa\circ\sigma^t_{\mathcal{M}_\kappa} 
\end{equation}
for every $\kappa \in K$, where $\sigma^t_{\mathcal{M}_\kappa}$ is the modular automorphism of $\mathcal{M}_\kappa$ associated with $\omega_\kappa$,
see \cite[Lemma 1]{Barnett95}\cite[Theorem 1]{Dykema94} (n.b.\! the proof of the latter essentially uses Voiculescu's computation of the $S$-operator in \cite[Lemma 1.8]{Voiculescu85}). This formula for the modular automorphisms immediately implies
that the modular operator $\Delta_\mathcal{M}$ and the modular conjugation $J_\mathcal{M}$ with respect to $\Omega$ are computed as 
\begin{equation}\label{Eq2.5} 
\Delta_\mathcal{M}^{it} = \bigstar_{\kappa \in K} \Delta_{\mathcal{M}_\kappa}^{it}, \quad 
J_\mathcal{M} = Z (\bigstar_{\kappa \in K} J_{\mathcal{M}_\kappa}) = (\bigstar_{\kappa \in K} J_{\mathcal{M}_\kappa}) Z, 
\end{equation}
where $\Delta_{\mathcal{M}_\kappa}$ and $J_{\mathcal{M}_\kappa}$ denote the modular operator and the modular conjugation for $\mathcal{M}_\kappa$ with respect to $\Omega_\kappa$.

\section{Free products of Half-sided modular inclusions}\label{hsmi}
Let $\mathcal{N}_\kappa \subset \mathcal{M}_\kappa$, $\kappa \in K$, be a (possibly infinite) family of inclusions of nontrivial
von Neumann algebras on Hilbert spaces $\mathcal{H}_\kappa$. Assume that each inclusion admits a common cyclic and separating unit vector $\Omega_\kappa \in \mathcal{H}_\kappa$. We denote by $\omega_\kappa$ the vector state obtained from $\Omega_\kappa$. 

\medskip
Let $(\mathcal{H},\Omega)$ be the free product of the $(\mathcal{H}_\kappa,\Omega_\kappa)$, $\kappa \in K$.
Consider the free products $(\mathcal{M},\omega) = \bigstar_{\kappa \in K} (\mathcal{M}_\kappa,\omega_\kappa)$
and $(\mathcal{N},\omega) = \bigstar_{\kappa \in K} (\mathcal{N}_\kappa,\omega_\kappa)$.
As explained in Section \ref{freeproduct}, these pairs are explicitly constructed as
$\mathcal{M} = \bigvee_{\kappa \in K} \lambda_\kappa(\mathcal{M}_\kappa) \supseteq \mathcal{N} = \bigvee_{\kappa \in K} \lambda_\kappa(\mathcal{N}_\kappa)$
and $\omega = \<\Omega, \cdot\,\Omega\>$ on $\mathcal{H}$ and $\Omega$ is a common cyclic and separating unit vector
for $\mathcal{N} \subseteq \mathcal{M}$. 

\medskip
As in \eqref{Eq2.4}, the modular automorphisms $\sigma^t_\mathcal{M}$ and $\sigma^t_\mathcal{N}$ of $\mathcal{M}$ and $\mathcal{N}$, respectively,
associated with $\omega$ are written as 
$$
\sigma^t_\mathcal{M}\circ\lambda_\kappa = \lambda_\kappa\circ\sigma^t_{\mathcal{M}_\kappa},
\quad 
\sigma^t_\mathcal{N}\circ\lambda_\kappa = \lambda_\kappa\circ\sigma^t_{\mathcal{N}_\kappa},   
$$
where $\sigma^t_{\mathcal{M}_\kappa}$ and $\sigma^t_{\mathcal{N}_\kappa}$ denotes the modular automorphisms of $\mathcal{M}_\kappa$ and $\mathcal{N}_\kappa$, respectively, associated with $\omega_\kappa$.
This description of the modular automorphisms immediately gives the next Lemma. 

\begin{lemma}\label{L1} If $\sigma^t_{\mathcal{M}_\kappa}(\mathcal{N}_\kappa) \subset \mathcal{N}_\kappa$ for all $t \geq 0$
for all $\kappa \in K$, then $\sigma^t_\mathcal{M}(\mathcal{N}) \subset \mathcal{N}$ for all $t \geq 0$. Hence, when all the $(\mathcal{N}_\kappa \subset \mathcal{M}_\kappa,\Omega_\kappa)$ are half-sided modular inclusions, so is $(\mathcal{N} \subset \mathcal{M}, \Omega)$.
\end{lemma} 

We call the triple $(\mathcal{N} \subset \mathcal{M},\Omega)$ the 
{\bf free product of the (given) half-sided modular inclusions} $(\mathcal{N}_\kappa \subset \mathcal{M}_\kappa,\Omega_\kappa)$, $\kappa \in K$, and write $(\mathcal{N} \subset \mathcal{M},\Omega) = \bigstar_{\kappa \in K} (\mathcal{N}_\kappa \subset \mathcal{M}_\kappa,\Omega_\kappa)$. 

\medskip
The computation of the modular automorphism $\sigma^t_\mathcal{M}$ also enables us to prove the following: 

\begin{lemma}\label{L2} If $\mathcal{N} = \mathcal{M}$, then $\mathcal{N}_\kappa = \mathcal{M}_\kappa$ for all $\kappa \in K$. 
\end{lemma}
\begin{proof}
Let us fix $\k \in K$.
By the formula of the modular automorphism $\sigma^t_\mathcal{N}$ explained above,
there exists a $\omega$-preserving conditional expectation $E$ from $\mathcal{N}$ onto $\lambda_\kappa(\mathcal{N}_\kappa)$.
By assumption, we have the restriction of $E$ to $\lambda_\kappa(\mathcal{M}_\kappa) \subset \mathcal{M} = \mathcal{N}$, which gives a faithful normal conditional expectation from $\lambda_\kappa(\mathcal{M}_\kappa)$ onto $\lambda_\kappa(\mathcal{N}_\kappa)$.
Hence $E_\k := \lambda_\kappa^{-1}\circ E \circ \lambda_\kappa$ gives a faithful normal conditional expectation from
$\mathcal{M}_\kappa$ onto $\mathcal{N}_\kappa$.
Observe that $\omega_\kappa(E_\k(x)) = \omega_\kappa(\lambda_\kappa^{-1}(E(\lambda_\kappa(x)))) = \omega(E(\lambda_\kappa(x))) = \omega(\lambda_\kappa(x)) = \omega_\kappa(x)$
for all $x \in \mathcal{M}_\kappa$. Hence $x\Omega_\kappa \in \mathcal{M}_\kappa \Omega_\kappa \mapsto E_\k(x)\Omega \in \mathcal{N}_\kappa\Omega_\kappa$ extends a unique orthogonal projection from $[\mathcal{M}_\kappa\Omega_\kappa]$ onto $[\mathcal{N}_\kappa\Omega_\kappa]$.  Since $\Omega_\kappa$ is cyclic for $\mathcal{N}_\kappa$, we conclude that $x\Omega_\kappa = E_\k(x)\Omega_\kappa$ holds for every $x \in \mathcal{M}_\kappa$.
Since $\Omega_\kappa$ is separating for $\mathcal{M}_\kappa$, we get $E_\k = \mathrm{id}$, that is, $\mathcal{N}_\kappa = \mathcal{M}_\kappa$. 
\end{proof} 

The above lemma actually says that if $\mathcal{N}_\kappa \subsetneqq \mathcal{M}_\kappa$ for some $\kappa \in K$, 
then $\mathcal{N} \subsetneqq \mathcal{M}$. We also remark that \emph{if the centralizers $(\mathcal{M}_\kappa)_{\omega_\kappa} = \mathbb{C}\1$
for all $\kappa \in K$, then $\mathcal{M}_\omega = \mathbb{C}\1$} thanks to \cite[Lemma 7]{Barnett95} (also see the proof of \cite[Lemma 2.1]{Ueda11-2}).

\medskip
In what follows, we denote by $J_{\mathcal{N}_\kappa}$ and $J_{\mathcal{M}_\kappa}$ the modular conjugation operators for $\mathcal{N}_\kappa$ and $\mathcal{M}_\kappa$, respectively,
with respect to $\Omega_\kappa$. We also denote by $J_\mathcal{N}$ and $J_\mathcal{M}$ the modular conjugation operators for $\mathcal{N}$ and $\mathcal{M}$, respectively, constructed by $\Omega$. Set $\Gamma_\kappa := J_{\mathcal{N}_\kappa} J_{\mathcal{M}_\kappa}$ for every $\kappa \in K$ as well as $\Gamma := J_\mathcal{N} J_\mathcal{M}$.
Furthermore, it holds that $\s_{\M_\k}^{4\pi n}(\N_\k) = \Ad \Gamma_\k^n(\N_\k)$
and $\s_{\M}^{4\pi n}(\N) = \Ad \Gamma^n(\N)$, see \cite[Corollary 4]{Wiesbrock93-1}

\begin{proposition}\label{pr:ergodicity} If $\Gamma_\kappa^n \to P_{\mathbb{C}\Omega_\kappa}$ weakly as $n\to\infty$ for all $\kappa \in K$, then $\Gamma^n \to P_{\mathbb{C}\Omega}$ weakly as $n\to\infty$.
In particular, if $\s_{\M_\k}^t$ is ergodic on on $\N_\k$ for every $\kappa \in K$,
then $\s_{\M}^t$ is ergodic on $\N$,
or equivalently, the canonical endomorphism $\gamma := \mathrm{Ad}\Gamma$
is ergodic,
i.e., $\bigcap_{n\geq1}\gamma^n(\mathcal{M}) = \mathbb{C}1$.
\end{proposition}
\begin{proof} By construction we observe that $J_{\mathcal{N}_\kappa}\Omega_\kappa = J_{\mathcal{M}_\kappa}\Omega_\kappa = \Omega_\kappa$, and thus $\Gamma_\kappa\Omega_\kappa = \Omega_\kappa$. Write $\Gamma_\kappa^\circ := \Gamma_\kappa\!\upharpoonright_{\mathcal{H}_\kappa^\circ}$ (which is well-defined) and observe by assumption that $(\Gamma_\kappa^\circ)^n \to 0$ weakly as $n\to\infty$. By the fact \eqref{Eq2.5}, we obtain that 
$$
\Gamma = (\bigstar_{\kappa \in K} J_{\mathcal{M}_\kappa})Z\,Z(\bigstar_{\kappa \in K} J_{\mathcal{N}_\kappa})
= \bigstar_{\kappa \in K} \Gamma_\kappa = \1 \oplus \bigoplus_{n\geq 1} \bigoplus_{\underset{1 \leq j \leq n-1}{\kappa_j \neq \kappa_{j+1}}} \Gamma_{\kappa_1}^\circ \otimes \cdots \otimes \Gamma_{\kappa_n}^\circ.  
$$
Since all the $\Gamma_\kappa^\circ$ are unitary operators, it is not hard to see that each $(\Gamma_{\kappa_1}^\circ \otimes \cdots \otimes \Gamma_{\kappa_n}^\circ)^n = (\Gamma_{\kappa_1}^\circ)^n \otimes \cdots \otimes (\Gamma_{\kappa_n}^\circ)^n \to 0$ weakly as $n\to\infty$, and thus $\Gamma^n \to P_{\mathbb{C}\Omega}$ weakly as $n\to\infty$. 

The latter assertion follows from the former thanks to \cite[Proposition 2.1, Corollary 2.2]{Longo84} and \cite[Remark 5.2]{Longo87}
and the fact that $\s_{\M}^{4\pi n}(\N) = \Ad \Gamma^n(\N)$ and $\Ad \Gamma(\M) \subset \N$.
\end{proof}
Let us remark that the fact that $\M$ is a factor follows from
\cite[Theorem 3.4]{Ueda11-1}, without the assumption of ergodicity.

One can also take the free product of the corresponding one-dimensional Borchers triples.
There one obtains an explicit construction of $T$ and the factoriality and the ergodicity follow
from the uniqueness of $\Omega$. As we will present the free product of M\"obius covariant nets in Section \ref{moebius}
and the construction will be similar, we omit it here.

The next proposition is a general assertion on free products of infinitely many copies of a fixed pair $(\mathcal{M}_0,\omega_0)$ of
von Neumann algebra and faithful normal state, and a nontrivial von Neumann subalgebra $\mathcal{N}_0 \subset \mathcal{M}_0$.
We emphasize that we do not assume that $\mathcal{N}_0$ is the range of a faithful normal conditional expectation from $\mathcal{M}_0$
(indeed, when $\M_0$ and $\N_0$ come from a local M\"obius covariant net, such a conditional expectation does not exist,
see Appendix \ref{noexpectation}).
One of the keys in the proof below is the explicit use of the conditional expectation $E_{K_1}$ for a subset $K_1$ of $K$
and also it is an important point in the proof below that the analytic approximation $b$ of $a$ need not stay in $\mathcal{N}$ in comparison to \cite[Theorem 5.2]{DM16}. 
The proof below is motivated by \cite[Theorem 3.3]{DDM14}. 
\begin{proposition}\label{P5}
Let $\mathcal{M}_0$ be a von Neumann algebra equipped with a faithful normal state $\omega_0$ and
$\mathcal{N}_0 \subset \mathcal{M}_0$ a nontrivial
von Neumann subalgebra.
Assume that all $(\mathcal{N}_\kappa \subset \mathcal{M}_\kappa, \omega_\kappa)$, $\kappa \in K$, are copies of
$(\mathcal{N}_0 \subset \mathcal{M}_0, \omega_0)$,
and moreover that the index set $K$ is infinite.
We set $\mathcal{N} := \bigvee_{\kappa \in K} \lambda_\kappa(\mathcal{N}_0) \subset \mathcal{M}$, where
$(\mathcal{M},\omega) = \bigstar_{\kappa \in K} (\mathcal{M}_\kappa,\omega_\kappa)$.
Then $\mathcal{N}' \cap \mathcal{M} = \mathbb{C}\1$, that is, the inclusion $\mathcal{N} \subset \mathcal{M}$ is irreducible. 
\end{proposition}
\begin{proof}
For any $x \in \mathcal{N}'\cap\mathcal{M}$ the new element $x - \omega(x)\1$
still stays in $\mathcal{N}'\cap \mathcal{M}$. Hence it suffices to prove that any $x \in \mathcal{N}'\cap\mathcal{M}$ with $\omega(x) = 0$ must be zero. 

For any subset $K_1 \subset K$, we define $\mathcal{M}_{K_1}$ is the von Neumann subalgebra generated by
$\{\lambda_\kappa(\mathcal{M}_0)\}_{\k \in K_1}$. By the definition of reduced free products, it is easy to see that $\mathcal{M}_{K_1}$ and $\mathcal{M}_{K\setminus K_1}$ are freely independent in $(\mathcal{M},\omega)$ (see \cite[Propositions 2.5.5(ii) and 2.5.7]{VDN92}).  Clearly, $\sigma^t_\mathcal{M}(\mathcal{M}_{K_1}) = \mathcal{M}_{K_1}$ for all $t \in \mathbb{R}$, and hence Takesaki's criterion guarantees that there exists a unique $\omega$-preserving conditional expectation $E_{K_1} : \mathcal{M} \to \mathcal{M}_{K_1}$. 

\medskip
Let $K_1 \subset K$ be an arbitrary {\it finite} subset. Then we can find ${\kappa_1'} \in K\setminus K_1$ since $K$ is an infinite set. Since $\mathcal{N}_0$ is nontrivial, one can also find a non-zero $a \in \mathcal{N}_0$ in such a way that 
$\omega_0(a) = 0$.

We claim the following:
for $x$ and $a$ as above, it holds that
\begin{align}\label{Claim6}
\lambda_{{\kappa_1'}}(a)E_{K_1}(x)\Omega \; \perp \; E_{K_1}(x)\lambda_{{\kappa_1'}}(a)\Omega.
\end{align}
Indeed, as remarked above, $\mathcal{M}_{K_1}$ and $\mathcal{M}_{K\setminus K_1}$ are freely independent in $(\mathcal{M},\omega)$
(see \eqref{eq:freeindep}). Remark that $\lambda_{{\kappa_1'}}(a) \in \mathcal{M}_{K\setminus K_1}$ and $\omega(\lambda_{{\kappa_1'}}(a)) = \omega_0(a) = 0$.
Observe that $\omega(E_{K_1}(x)) = \omega(x) = 0$. Hence we have 
\begin{align*}
\<\lambda_{{\kappa_1'}}(a)E_{K_1}(x)\Omega,\; E_{K_1}(x)\lambda_{{\kappa_1'}}(a)\Omega\>  
&= 
\<\Omega,\; E_{K_1}(x^*)\lambda_{{\kappa_1'}}(a^*) E_{K_1}(x)\lambda_{{\kappa_1'}}(a)\Omega\> \\
&= 
\omega(E_{K_1}(x^*)\lambda_{{\kappa_1'}}(a^*) E_{K_1}(x)\lambda_{{\kappa_1'}}(a)) = 0
\end{align*}
by free independence. 

Let $b \in \mathcal{M}_0$ be an arbitrary, $\sigma_{\mathcal{M}_0}$-analytic element.
Then, as $x \in \N'\cap \M$, we have 
\begin{align*}
0 
&= 
\lambda_{{\kappa_1'}}(a)x - x\lambda_{{\kappa_1'}}(a) \\
&= 
\lambda_{{\kappa_1'}}(a)(x-E_{K_1}(x)) + \lambda_{{\kappa_1'}}(a)E_{K_1}(x) - (x - E_{K_1}(x))\lambda_{{\kappa_1'}}(a) - E_{K_1}(x)\lambda_{{\kappa_1'}}(a) \\
&=
\lambda_{{\kappa_1'}}(a)(x-E_{K_1}(x)) + \lambda_{{\kappa_1'}}(a)E_{K_1}(x)  \\ 
&\qquad\qquad\qquad - (x - E_{K_1}(x))\lambda_{{\kappa_1'}}(a-b) - (x - E_{K_1}(x))\lambda_{{\kappa_1'}}(b) -  E_{K_1}(x)\lambda_{{\kappa_1'}}(a)
\end{align*} 
and hence 
\begin{align*}
&E_{K_1}(x)\lambda_{{\kappa_1'}}(a)\Omega - \lambda_{{\kappa_1'}}(a)E_{K_1}(x)\Omega \\
&\qquad= 
\lambda_{{\kappa_1'}}(a)(x-E_{K_1}(x))\Omega - (x - E_{K_1}(x))\lambda_{{\kappa_1'}}(a-b)\Omega - (x - E_{K_1}(x))\lambda_{{\kappa_1'}}(b)\Omega. 
\end{align*}
By \eqref{Claim6}, the vectors on the left-hand side are orthogonal. Thus we obtain that  
$$
\Vert \lambda_{{\kappa_1'}}(a)E_{K_1}(x)\Omega\Vert 
\leq 
\Vert a \Vert\,\Vert (x-E_{K_1}(x))\Omega\Vert + 2\Vert x\Vert\,\Vert \lambda_{{\kappa_1'}}(a-b)\Omega\Vert +  \Vert(x - E_{K_1}(x))\lambda_{{\kappa_1'}}(b)\Omega\Vert. 
$$ 
Since $b$ is $\sigma_{\mathcal{M}_0}$-analytic, $\lambda_{{\kappa_1'}}(b)$ is also $\sigma_\mathcal{M}$-analytic and  
$$
(x - E_{K_1}(x))\lambda_{{\kappa_1'}}(b)\Omega = J_\mathcal{M}\lambda_{{\kappa_1'}}(\sigma^{\mathrm{i}/2}_{\mathcal{M}_0}(b)^*)J_\mathcal{M}(x-E_{K_1}(x))\Omega. 
$$ 
Therefore, we get 
\begin{align*}
&\Vert \lambda_{{\kappa_1'}}(a)E_{K_1}(x)\Omega\Vert \\
&\qquad\leq 
\Vert a \Vert\,\Vert (x-E_{K_1}(x))\Omega\Vert + 2\Vert x\Vert\,\Vert \lambda_{{\kappa_1'}}(a-b)\Omega\Vert 
+  \Vert \sigma^{\mathrm{i}/2}_{\mathcal{M}_0}(b)\Vert\,\Vert(x - E_{K_1}(x))\Omega\Vert.
\end{align*}
Observe that 
$$
\Vert \lambda_{{\kappa_1'}}(a-b)\Omega\Vert^2 
= 
\omega(\lambda_{{\kappa_1'}}((a-b)^* (a-b))) = \omega_0((a-b)^* (a-b)) = \Vert (a-b)\Omega_0\Vert^2
$$
and that 
\begin{align*} 
\Vert \lambda_{{\kappa_1'}}(a)E_{K_1}(x)\Omega\Vert^2 
&= 
\omega(E_{K_1}(x)^* \lambda_{{\kappa_1'}}(a^* a) E_{K_1}(x)) \\
&= 
\omega(\lambda_{{\kappa_1'}}(a^* a)) \cdot\omega(E_{K_1}(x)^* E_{K_1}(x)) \qquad \text{(by free independence)} \\
&=
\omega_0(a^* a) \cdot\omega(E_{K_1}(x)^* E_{K_1}(x)) \\
&= 
\Vert a\Omega_0\Vert^2\,\Vert E_{K_1}(x)\Omega\Vert^2. 
\end{align*} 
Consequently, we have obtained that 
$$
\Vert a\Omega_0\Vert\,\Vert E_{K_1}(x)\Omega\Vert 
\leq 
\Vert a \Vert\,\Vert (x-E_{K_1}(x))\Omega\Vert + 2\Vert x\Vert\,\Vert (a-b)\Omega_0\Vert + \Vert \sigma^{\mathrm{i}/2}_{\mathcal{M}_0}(b)\Vert\,\Vert(x - E_{K_1}(x))\Omega\Vert.
$$
Taking the limit as $K_1 \nearrow K$ we get 
\begin{equation}\label{Eq3.1}
\Vert a\Omega_0\Vert\,\Vert x\Omega\Vert 
\leq 
2\Vert x\Vert\,\Vert (a-b)\Omega_0\Vert, 
\end{equation}
because the Jones projection $e_{K_1}$ associated with $E_{K_1}$ converges to the identity operator thanks to $\mathcal{M}_{K_1} \nearrow \mathcal{M}$ as $K_1 \nearrow K$. One can easily see that the linear manifold of all $b\Omega_0$ of $\sigma_{\mathcal{M}_0}$-analytic elements $b \in \mathcal{M}_0$ is dense in $\mathcal{H}_0$. Hence the right-hand side of \eqref{Eq3.1} can arbitrarily be small, since $b$ is arbitrary. Consequently, we get $x = 0$, since $a$ is non-zero and also since $\Omega_0$ and $\Omega$ are separating for $\mathcal{M}_0$ and $\mathcal{M}$, respectively.
\end{proof}  

The next theorem is one of the main results of this paper, which is now a simple
consequence of Proposition \ref{P5}.

\begin{theorem}\label{T6}
Let $(\mathcal{N}_0 \subset \mathcal{M}_0, \Omega_0)$ be a half-sided modular inclusion
with ergodic canonical endomorphism.
Then the free product $(\mathcal{N} \subset \mathcal{M}, \Omega)$
of 
infinitely many copies of $(\mathcal{N}_0 \subset \mathcal{M}_0, \Omega_0)$
is a standard half-sided modular inclusion with ergodic canonical endomorphism
and trivial relative commutant.
\end{theorem}
\begin{proof}
This is a half-sided modular inclusion due to Lemma \ref{L1}.
By Proposition \ref{pr:ergodicity} the canonical endomorphism of the resulting half-sided modular inclusion
$(\mathcal{N} \subset \mathcal{M}, \Omega)$ is ergodic.
Finally, Proposition \ref{P5} shows that $\mathcal{N}'\cap\mathcal{M} = \mathbb{C}\1$.
\end{proof}

In this Theorem and Proposition \ref{P5}, it is important that we take copies of the \emph{identical}
inclusion $\N_0 \subset \M_0$. If we remove this condition, we find indeed
an infinite family $(\N_\k \subset \M_\k, \Omega_\k)$ whose free product has nontrivial relative commutant
in Proposition \ref{th:infinite}.
Such an infinite family must be a little elaborated, and it is not enough that
the relative commutants $\N_\k'\cap \M_\k$ get larger as $\k \to \infty$.
For example, if one just take a one-dimensional Borchers triple $(\M_0, T_0, \Omega_0)$ and
consider the shifted family $(\Ad T_0(a_\k)(\M_0) \subset \M_0, \Omega_0)$ where $a_\k \to \infty$,
all such inclusions are actually unitarily equivalent to a fixed one $(\Ad T_0(1)(\M_0) \subset \M_0, \Omega_0)$
(through unitaries $\Delta_{\M_0}^{it_\k}$ with appropriate $t_\k$, which preserve the vacuum $\Omega_0$)
and their free product has the trivial relative commutant by Theorem \ref{T6}.
The example in Proposition \ref{th:infinite} avoids this problem by taking an explicitly
non equivalent family. On the other hand, we do not know whether there is a free product
HSMI with nontrivial relative commutant, even if the family is finite.

\section{Free products of M\"obius covariant nets}\label{moebius}
In this Section, we exploit the techniques in M\"obius covariant nets
in order to show that inclusions of free product von Neumann algebras
may have nontrivial relative commutant.

Let $\{(\A_\k, U_\k, \Omega_\k)\}_{\k\in K}$ be a family of M\"obius covariant nets,
where $K$ is an index set.
Thanks to the Reeh-Schlieder theorem \cite[Theorem 2.1(i)]{DLR01}, the vacuum vector $\Omega_\kappa$ is a common cyclic and separating vector for all $\mathcal{A}_\kappa(I)$, $I \in \mathcal{I}$. Let $(\mathcal{H},\Omega) = \bigstar_{\kappa \in K} (\mathcal{H}_\kappa,\Omega_\kappa)$.
As in Section \ref{freeproduct}, the free products $(\mathcal{A}(I),\omega) = \bigstar_{\kappa \in K} (\mathcal{A}_\kappa(I),\omega_\kappa)$, $I \in \mathcal{I}$, are simultaneously constructed as 
$$
\mathcal{A}(I) := \bigvee_{\kappa \in K} \lambda_\kappa(\mathcal{A}_\kappa(I)), \quad \omega = \<\Omega, \cdot\,\Omega\> 
$$
in $\mathcal{H}$. By construction, $\mathcal{A} := (\mathcal{A}(I))_{I \in \mathcal{I}}$ satisfies the isotony.
We remark that $\Omega$ is again cyclic and separating for all $\mathcal{A}(I)$, $I \in \mathcal{I}$,
that is, $\Omega$ plays a role of the vacuum vector for the family $\mathcal{A}$.
Let us consider $U(g) := \bigstar_{\kappa \in K} U_\kappa(g)$ for every $g \in G$,
which is well-defined since $U_\kappa(g)\Omega_\kappa = \Omega_\kappa$ and $U_\k$ is a unitary representation.
It is not hard to see that $g \mapsto U(g)$ is a unitary representation on $\mathcal{H}$ and fixes $\Omega$. Here is an expected fact.  

\begin{proposition}\label{pr:moebius}
The resulting $(\mathcal{A} = (\mathcal{A}(I))_{I \in \mathcal{I}}, U, \Omega)$ is a M\"obius covariant net.
It fails to have locality if at least two of the nets $\A_\k$ are nontrivial.
If all the given $\mathcal{A}_\kappa$ are local, then $\mathcal{A}$ is twisted local: namely,
the unitary involution $Z$ in Section \ref{freeproduct} commutes with $U(g)$ and satisfies $Z\A(I)Z = \A(I')'$.
\end{proposition}
\begin{proof}
$\mathcal{A}$ trivially satisfies the isotony. As remarked above, $\Omega$ plays a role of the vacuum vector. 
If two of them are nontrivial, say $\k_1, \k_2$, then for nontrivial elements $x_1 \in \A_{k_1}(I_1), x_2 \in \A_{\k_2}(I_2)$,
$\l_{\k_1}(x_1)$ and $\l_{\k_2}(x_2)$ do not commute, therefore, locality is lost. 

By the definition of the representations $\lambda_\kappa$ and the construction of $U(g)$,
it is plain to see that $U(g)\lambda_\kappa(x)U(g)^* = \lambda_\kappa(U_\kappa(g)x U_\kappa(g)^*)$
for every $x \in \mathcal{M}_\kappa$, $\kappa \in K$ and $g \in G$. Hence we obtain that 
$$
U(g)\mathcal{A}(I)U(g)^* = 
\bigvee_{\kappa \in K} \lambda_\kappa(U_\kappa(g)\mathcal{A}_\kappa(I)U_\kappa(g)^*) = \bigvee_{\kappa \in K} \lambda_\kappa(\mathcal{A}_\kappa(gI)) = \mathcal{A}(gI)
$$
for every $g \in G$ and $I \in \mathcal{I}$. Thus we have confirmed that $\mathcal{A}$ satisfies the M\"obius covariance. 

Since the representation $U_\kappa$ fixes $\Omega_\kappa$, we have $L_{0,\kappa}\Omega_\kappa = 0$ and thus may and do think of $L_{0,\kappa}$
as a positive self-adjoint operator on $\mathcal{H}_\kappa^\circ$. By the construction of $U(g)$ we see 
\begin{align*}
&\lim_{\theta\to0} \frac{1}{\mathrm{i}\theta} \Big(U(R(\theta))(\xi_1\otimes \cdots \otimes \xi_n) - (\xi_1\otimes \cdots \otimes \xi_n)\Big) \\
&\qquad= 
(L_{0,\kappa_1}\xi_1)\otimes\xi_2\otimes\cdots\otimes\xi_n + \xi_1\otimes(L_{0,\kappa_2}\xi_2)\otimes\cdots\otimes\xi_n+ \cdots
\end{align*} 
for every $\xi_1\otimes\cdots\otimes\xi_n \in \mathcal{H}_{\kappa_1}^\circ\,\bar{\otimes}\,\cdots\,\bar{\otimes}\,\mathcal{H}_{\kappa_n}^\circ \subset \mathcal{H}$
with $\xi_j \in \mathscr{D}(L_{0,\kappa_j})$, the domain of $L_{0,\kappa_j}$.
Consequently, the conformal Hamiltonian $L_0$, i.e., the generator of $\theta \mapsto U(R(\theta))$, acts on the dense subspace of $\mathcal{H}_{\kappa_1}^\circ\,\bar{\otimes}\,\cdots\,\bar{\otimes}\,\mathcal{H}_{\kappa_n}^\circ$ (algebraically) spanned by simple tensors $\xi_1\otimes\cdots\otimes\xi_n$ with $\xi_j \in \mathscr{D}(L_{0,\kappa_j})$
and becomes 
$$
L_{0,\kappa_1}\otimes \1 \otimes\cdots\otimes \1  + \1 \otimes L_{0,\kappa_2} \otimes \1 \otimes \cdots \otimes \1 + \cdots + \1 \otimes \cdots \1 \otimes L_{0,\kappa_n}, 
$$
which is an essentially self-adjoint positive operator on
$ \mathcal{H}_{\kappa_1}^\circ\,\bar{\otimes}\,\cdots\,\bar{\otimes}\,\mathcal{H}_{\kappa_n}^\circ \subset \mathcal{H}$
(see \cite[Corollary of Theorem V.III.33]{RSI}). Then it immediately follows that $L_0$ is positive.
Therefore, we have confirmed that $\mathcal{A}$ satisfies the positivity of the energy.
From this expression of $L_0$, it is also clear that $\Omega$ is the unique (up to scalar) invariant vector
under $e^{i\theta L_0}$.

Finally, the existence of the vacuum is immediate by the commutation relation for free products \eqref{Eq2.3}, since we have known that $\Omega_\kappa$ is a common cyclic and separating vector for all $\mathcal{A}_\kappa(I)$, $I \in \mathcal{I}$. Hence we have proved the first part. 

\medskip
Assume that all $\mathcal{A}_\kappa$ are local.
We first remark that $Z$ commutes with $\Delta_I^{it} := \Delta_{\mathcal{A}(I)}^{it}$ and
$U(g)$ for all $I \in \mathcal{I}$, $t \in \mathbb{R}$ and $g \in G$.
Thus it suffices, for the twisted locality, to confirm that $\mathcal{A}(I')' = Z\mathcal{A}(I)Z$ for all $I \in \mathcal{I}$. By the locality of the $\mathcal{A}_\kappa$, $\kappa \in K$, we have $\mathcal{A}_\kappa(I') = \mathcal{A}_\kappa(I)'$ on $\mathcal{H}_\kappa$ (Haag duality) for all $\kappa \in K$ and $I \in \mathcal{I}$, see e.g.\! \cite[Proposition 2.9]{DLR01}. Hence, by the facts \eqref{Eq2.2},\eqref{Eq2.3} for free products, we obtain that
\begin{equation}\label{Eq4.1} 
\begin{aligned}
\mathcal{A}(I')' 
&= 
\Big(\bigvee_{\kappa \in K} \lambda_\kappa(\mathcal{A}_\kappa(I'))\Big)' =  \bigvee_{\kappa \in K} \rho_\kappa(\mathcal{A}_\kappa(I')') \\
&= 
\bigvee_{\kappa \in K} \rho_\kappa(\mathcal{A}_\kappa(I)) = \bigvee_{\kappa \in K} Z\lambda_\kappa(\mathcal{A}_\kappa(I))Z = Z\mathcal{A}(I)Z.
\end{aligned}
\end{equation} 
Hence we are done. 
\end{proof}  

We may write 
$$
(\mathcal{A},
U,\Omega) = \bigstar_{\kappa \in K} (\mathcal{A}_\kappa,
U_\kappa.\Omega_\kappa) \quad \text{or simply} \quad 
\mathcal{A} = \bigstar_{\kappa \in K} \mathcal{A}_\kappa
$$ 
and call it the \emph{free product M\"obius covariant net} of the given
$(\mathcal{A}_\kappa,
U_\kappa,\Omega_\kappa)$, $\kappa \in K$. The next assertion is a simple application of \cite[Corollary B]{HU16}. 

\begin{proposition}\label{P9} 
For every $\kappa \in K$ such that $\H_\k$ is not one-dimensional,
there is no {\it local} subnet \footnote{A subnet $\B$ of $(\A, U, \Omega)$ is a family of von Neumann subalgebras $(\B(I))_{I\in \I}$ with $\B(I) \subset \A(I)$
satisfying the isotony and the M\"obius covariance with respect to the same $U$.}
of $\mathcal{A}$ with the split property which is larger than $\mathcal{A}_\kappa$. 
\end{proposition}
\begin{proof}
Assume that $\mathcal{B} = (\mathcal{B}(I))_{I\in\mathcal{I}}$ is a subnet with the split property such that
$\mathcal{A}_\kappa(I) \subseteq \mathcal{B}(I) \subseteq \mathcal{A}(I)$ for every $I \in \mathcal{I}$.
Thanks to e.g.\! \cite[Proposition 3.1]{DLR01}, $\mathcal{B}(I)$ must be hyperfinite by the split property
and $\A_\k(I)$ is of type III$_1$ by \cite[Proposition 2.4(ii)]{DLR01}, in particular diffuse,
and there is a faithful normal conditional expectation from $\A(I)$ onto $\mathcal{B}(I)$ by the Bisognano-Wichmann property 
(n.b.\! the M\"obius covariance of $\mathcal{B}$ is provided by the same $U$)
and Takesaki's theorem \cite[Theorem IX.4.2]{TakesakiII}. 
Hence \cite[Corollary B]{HU16} shows that $\mathcal{B}(I)$ must sit inside $\mathcal{A}_\kappa(I)$.
\end{proof}

We conjecture that the above statement still holds under assuming only locality on subnets. 

\medskip
Let us consider the dual M\"obius covariant net $\widehat{\mathcal{A}} = (\widehat{\mathcal{A}}(I))_{I \in \mathcal{I}}$ defined by $\widehat{\mathcal{A}}(I) := \mathcal{A}(I')' $ in $\mathcal{H}$, see \cite[\S2]{DLR01}. If all the given $\mathcal{A}_\kappa$, $\kappa \in K$ are local, then 
$$
\widehat{\mathcal{A}}(I) = \mathcal{A}(I')' = Z\mathcal{A}(I)Z = \bigvee_{\kappa \in K} \rho_\kappa(\mathcal{A}_\kappa(I)), 
$$
as in \eqref{Eq4.1}. It is known that {\bf the observable net} $\C = (\C(I))_{I \in \mathcal{I}}$ with $\C(I) := \mathcal{A}(I) \cap \widehat{\mathcal{A}}(I)$ is local (see \cite[Remark 2.8]{DLR01}). However, the observable net turns out to be trivial under a natural requirement.

\begin{proposition}\label{P10} 
If all the $\mathcal{A}_\kappa$ are local and two of them are nontrivial,
then $\mathcal{C}(I) = \mathbb{C}\1$ for all $I \in \mathcal{I}$.    
\end{proposition}
\begin{proof} As we assume the uniqueness of the vacuum,
by e.g.\! \cite[Propositions 2.3(ii) and 2.5]{DLR01} there exist $\kappa_1 \neq \kappa_2 \in K$ such that both $\mathcal{A}_{\kappa_j}(I)$, $j=1,2$, are type III$_1$ factors for every $I \in \mathcal{I}$. Let $I \in \mathcal{I}$ be arbitrarily fixed. Since $\mathcal{A}_{\kappa_j}(I')$ is a type III$_1$ factor, we can choose an isometry $s_j \in \mathcal{A}_{\kappa_j}(I')$ such that $s_j^n (s_j^n)^* \to 0$ weakly as $n\to\infty$. For each $j=1,2$ we observe that 
$$
[\lambda_{\kappa_j}(s_j^n)^*, \rho_{\kappa_j}(s_j^n)] = 
(P_{\mathbb{C}\Omega} + P_{\mathcal{H}_{\kappa_j}^\circ})\lambda_{\kappa_j}(1 - s_j^n (s_j^n)^*) \to P_{\mathbb{C}\Omega} + P_{\mathcal{H}_{\kappa_j}^\circ}
$$
weakly as $n\to\infty$. Hence $P_{\mathbb{C}\Omega} + P_{\mathcal{H}_{\kappa_j}^\circ}$ falls into $\mathcal{A}(I') \vee \widehat{A}(I')$ for every $j=1,2$. Therefore, we have 
$$
P_{\mathbb{C}\Omega} =  
(P_{\mathbb{C}\Omega} + P_{\mathcal{H}_{\kappa_1}^\circ})(P_{\mathbb{C}\Omega} + P_{\mathcal{H}_{\kappa_2}^\circ}) \in \mathcal{A}(I') \vee \widehat{A}(I') = \mathcal{C}(I)', 
$$
from which $\mathcal{C}(I)'$ must be $B(\mathcal{H})$ since $\Omega$ is cyclic for it. Hence we are done.   
\end{proof}
Actually, we need only that there exist $\kappa_1\neq\kappa_2 \in K$ such that both $\mathcal{A}_{\kappa_j}(I)$, $j = 1,2$,
are properly infinite for every $I \in \mathcal{I}$ and the uniqueness of the vacuum (the factoriality of $\A_\k(I)$)
is not essential.

Now let us move to the existence of inclusions of free product von Neumann algebras with nontrivial relative commutant,
which is our second main result.
We take finitely many M\"obius covariant nets $\{\A_\k\}_{\k \in K}$ with trace class property for all $s > 0$
and make the free product $\bigstar_{\k\in K} \A_\k$ of them, which we denote by $\A$.
We lose locality through free product, but the trace class property for sufficiently large $s$
survives.

\begin{theorem}\label{th:distalnet}
 If $K$ is finite, and if  M\"obius covariant nets $\A_\k$ satisfy the trace class property,
 then the free product net $\A$ satisfies the distal split property.
 In particular, for a pair $I\subset \tilde I$ such that the split inclusion holds,
 the relative commutant of the inclusion $\A(I) \subset \A(\tilde I)$ of free product von Neumann algebras is nontrivial. Moreover, $\A(I)' \cap \A(\tilde I)$ must be
 a type III$_1$ factor.
\end{theorem}
\begin{proof}
 It is enough to show that $\Tr e^{-sL_0} < \infty$ for some $s$.
 We take $s$ such that $\Tr e^{-sL^\circ_{0,\k}} < \epsilon < \frac 1{|K|-1}$, where $L^\circ_{0,\k}$ is the restriction of $L_{0,\k}$ on $\H^\circ_\k$:
 note that $L_{0,\k}$ has the eigenvalue $0$ with multiplicity $1$ by the uniqueness of the vacuum,
 while the contributions from the other eigenvalues can arbitrarily be small as $s$ gets larger.
 And note that $K$ is finite, hence for a sufficiently large $s$ this condition is satisfied for all $\k \in K$.
 
 Recall that the free product Hilbert space is given by
\begin{align*}
\mathcal{H} &= \mathbb{C}\Omega \oplus \bigoplus_{n\geq 1} \bigoplus_{\underset{1 \leq j \leq n-1}{\kappa_j \neq \kappa_{j+1}}} \mathcal{H}^\circ_{\kappa_1}\botimes\cdots\botimes\mathcal{H}^\circ_{\kappa_n} \quad \text{with} \quad \mathcal{H}^\circ_\kappa := \mathcal{H}_\kappa\ominus\mathbb{C}\Omega_\kappa. 
\end{align*}
 Note that, for a fixed $n$, the Hilbert space is a direct sum of $|K|\cdot (|K|-1)^{n-1}$ tensor products.
 Accordingly, we can estimate
 \[
  \Tr e^{-sL_0} \le 1 + \sum_n \epsilon^n |K|\cdot (|K|-1)^{n-1} = 1 + \frac{|K|}{|K|-1}\sum_n  \epsilon^n(|K|-1)^{n} < \infty,
 \]
 as we chose $\epsilon < \frac1{|K|-1}$.

It follows by \cite[Corollary 6.4]{BDL07} that $\A(I) \subset \A(\tilde I)$ is split,
hence its relative commutant is again a nontrivial type III$_1$ factor.
\end{proof}

\begin{theorem}\label{th:infinite}
 Let $(\A_0, U_0, \Omega_0)$ be a M\"obius covariant net with the trace class property, let $K$ be a countably infinite set.
 We fix an interval $\tilde I$ and take a sequence $I_\k \subset \tilde I$
 such that the inner distances  \footnote{The inner distance $\ell(\tilde I, I_\k)$ can arbitrarily be large if $I_\k$ is sufficiently small,
 see \cite[Section 3.1]{BDL07} and $L_{0,0}$ has the eigenvalue $0$ with multiplicity $1$,
 hence such a sequence $s_\k$ exists.}
 satisfy $\ell(\tilde I, I_\k) > s_\k$ with 
 $\sum_\k (\Tr e^{-s_\k L_{0,0}}-1) < 1$, where $L_{0,0}$ is the conformal Hamiltonian for the net $\A_0$.
 Then the inclusion $\N \subset \M$ of the free product von Neumann algebras
 $\N = \bigstar_{\k \in K}(\A_0(I_\k), \Omega_0), \M = \bigstar_{\k \in K}(\A_0(\tilde I), \Omega_0)$ has nontrivial relative commutant.
\end{theorem}
\begin{proof}
 This time we invoke the $L^2$-nuclearity condition, namely, we are going to prove
 that $T_{\M,\N} := \Delta_\M^\frac14\Delta_\N^{-\frac14}$ is of trace class,
 then the inclusion $\N \subset \M$ is split \cite[Propositions 5.2, 5.3]{BDL07}.
 
 Let $\Delta_{0,\tilde I}, \Delta_{0,I_\k}$ be the modular operators of $\A_0(\tilde I), \A_0(I_\k)$ with respect to $\Omega_0$,
 respectively. As both $\Delta_{0,\tilde I}^{it}, \Delta_{0,I_\k}^{it}$ leave $\Omega_0$ invariant,
 we can naturally define $\Delta_{0,\tilde I}^\circ, \Delta_{0,I_\k}^\circ$, their restriction on $\H_0^\circ$.
 We know from \cite[Proposition 3.1]{BDL07} that
 $\|\Delta_{0,\tilde I}^{\frac14}\Delta_{0,I_\k}^{-\frac14}\|_1 = \|e^{-s_\k L_{0,0}}\|_1$,
 where $\|\cdot\|_1$ denotes the trace norm.
 Again, the eigenvalue $0$ of $L_{0,0}$ has multiplicity $1$, and by assumption
 $1 > \sum_\k (\|e^{-s_\k L_{0,0}}\|_1 - 1) = \sum_\k (\Tr e^{-s_\k L_{0,0}} - 1)
 = \sum_\k (\|(\Delta_{0,\tilde I}^\circ)^\frac14(\Delta_{0,I_\k}^\circ)^{-\frac14}\|_1) =: \epsilon$.
 
  From their componentwise expressions, we know what $T_{\M, \N}$ looks like:
 \[
  T_{\M,\N} = \Delta_\M^\frac14\Delta_\N^{-\frac14}
  = \1 \oplus \bigoplus_{n\geq 1} \bigoplus_{\underset{1 \leq j \leq n-1}{\kappa_j \neq \kappa_{j+1}}}
  (\Delta_{0,\tilde I}^\circ)^\frac14(\Delta_{0, I_{\k_1}}^\circ)^{-\frac14}\botimes\cdots\botimes(\Delta_{0,\tilde I}^\circ)^\frac14(\Delta_{0, I_{\k_n}}^\circ)^{-\frac14}
 \]
 Its trace norm can be estimated recursively in $n$:
 the contribution from $n=1$ is exactly $\epsilon$ above.
 The contribution from $n+1$ can be dominated by the contribution from $n$
 multiplied by $\sum_\k (\Delta_{0,\tilde I}^\circ)^\frac14(\Delta_{0, I_{\k}}^\circ)^{-\frac14}$, say, from the right
 (namely, any sequence $\k_1,\cdots \k_{n+1}$ can be obtained from a sequence $\k_1,\cdots \k_n$ by adding an element $\k_{n+1}$
 to the right).
 Therefore, the total sum is less than $1 + \sum_n \epsilon^n = \frac1{1-\epsilon} < \infty$.
\end{proof}

We emphasize that these Theorems contrast quite sharply with Proposition \ref{P5}. 

\section{Free products of Borchers triples}\label{borchers}

As we recalled, to any two-dimensional Haag-Kastler net, one can associate
a two-dimensional Borchers triple. The simplest one is the so-called the free field net
which we will review below, but there are more important interacting nets
ranging from the constructive QFT \cite{GJ87} to the recent operator-algebraic constructions
\cite{Lechner08, Tanimoto14-1, AL16}.
Note also that there are two-dimensional nets with trivial double cone algebras up to a given size \cite{LL15},
although they are subnets of the free field net.

Given a family of Borchers triples, one can promote it to
the free product. By arguments parallel to those of Proposition \ref{pr:moebius} and Theorem \ref{T6},
it is straightforward to show the following.
\begin{proposition}\label{pr:borchers}
Let $\{(\M_\k, T_\k, \Omega_\k)\}_{\k \in \K}$ be a family of two-dimensional Borchers triples on Hilbert spaces $\{\H_\k\}_{\k\in K}$.
Set
\[
(\H,\Omega) := \bigstar_{\k\in K} (\H_\k,\Omega_\k), \quad 
(\M,\omega) := \bigstar_{\k\in K} (\M_\k,\omega_\k) \quad
 T(a) := \bigstar_{\k\in K} T_\k(a). 
\]
Then $(\M,T,\Omega)$ is again a two-dimensional Borchers triple.
Furthermore, if $K$ is an infinite index set, and all these Borchers triples are identical,
then the relative commutant $\Ad T(a)(\M)'\cap \M$ is trivial for any $a \in W_{\mathrm R}$.
\end{proposition}

Yet, if the index set $K$ is finite, or if $(\M_\k, T_\k, \Omega_\k)$'s are not identical,
$\Ad T(a)(\M)'\cap \M$ might be nontrivial, and if so, one could construct a Haag-Kastler net.

\paragraph{Two-dimensional Borchers triples of the massive free field.}
In order to study closely the (possible) physical properties of the resulting Borchers triple,
we present explicitly the simplest Borchers triple, coming from the free massive field.
Our notations follow those of \cite{Lechner03} (except for $z(\psi)$, which is antilinear in $\psi$ in our notation).
The one-particle Hilbert space with the mass $m>0$
is given by $\H_1 := L^2(\RR, d\theta)$ and 
the translation group acts by $(T_1(a)\psi)(\theta) = e^{i p_m(\theta)\cdot a}\psi(\theta)$,
where $p_m(\theta) := (m\cosh\theta, m\sinh\theta)$,
and for the Lorentz metric we use the convention $a\cdot b = a_0b_0 - a_1b_1$.
We introduce
the symmetrized Hilbert space $\H := \bigoplus P_n \H_1^{\otimes n}$,
where $P_n$ is the projection onto the symmetric subspace.

Let $z^\dagger$ and $z$ be the creation and annihilation operator,
namely,
$(z^\dagger(\psi)\Psi)_n = \sqrt{n}P_n(\psi\otimes \Psi_{n-1})$
for vectors with finite particle number $\Psi = (\Psi_n), \Psi_n \in P_n\H^{\otimes n}$.
The annihilation operator is the adjoint $z(\psi) = z^\dagger(\psi)^*$.
The (real) free field $\phi$ is defined by
\[
\phi(f) := z^\dagger(f^+) + z(J_1f^-), \qquad f^\pm(\theta) = \frac{1}{2\pi}\int d^2a\,f(a)e^{\pm i p_m(\theta)\cdot a},
\]
where $f$ is a test function of the Schwartz class $\mathscr{S}(\RR^2)$ and $J_1\psi(\theta) = \overline{\psi(\theta)}$.
Our von Neumann algebra is
\[
\M := \{e^{i\phi(f)}: \supp f \subset W_{\mathrm R}\}''.
\]
The translation on the full space is the second quantized representation
$T(a) := \Gamma(T_1(a))$ and there is the Fock vacuum vector $\Omega \in \H$.
Then $(\M, T, \Omega)$ is a Borchers triple and satisfies the modular nuclearity condition
for any $a \in W_{\mathrm R}$ \cite[Section 4]{BL04}.

\paragraph{Attempts towards the intersection.}
Let $(\M_\k, T_\k, \Omega_\k), \k = 1,2$ be two copies of the free field Borchers triple.
We will put the index $\k$ to all the objects to distinguish these copies.
We saw in Section \ref{moebius} that the trace class condition and the $L^2$-nuclearity condition in M\"obius covariant nets
are preserved under finite free products, at the cost of a finite distance.
Let us take two copies of the massive free field and consider their free product.
Below are some standard ways to prove the nontriviality of the intersection
$\Ad T(a)(\M)'\cap \M$, which unfortunately do not work.
\begin{itemize}
 \item The trace class property does not make sense in two dimensions.
 \item The $L^2$-nuclearity condition amounts to compute the operator $\Delta_\M^\frac14 \Delta_{\N}^{-\frac14}$,
 where $\N = \Ad T(a)(\M)$.
In the two-dimensional situation this is a multiplication operator (of $\theta$'s), hence it cannot be nuclear.
 \item At first sight, there could be two ways to attempt to prove modular nuclearity.
 One is, by noting that the tensor product net satisfies the modular nuclearity condition,
 to consider the map $\lambda_{\k_1}(x_1)\lambda_{\k_2}(x_2)\cdots\lambda_{\k_n}(x_n) \longmapsto x_1\otimes x_2\otimes\cdots \otimes x_n$.
 We are unable to carry this out because the map cannot be shown to be continuous in norm.
 The other is to mimic the proof of \cite{Lechner08}. This is obstructed as there is no symmetry
 in, e.g.\! vectors such as $\lambda_{\k_1}(z^\dagger_1(f^+))\lambda_{\k_2}(z_2^\dagger(g^+))\Omega$ considered as a two-variable function, and the core analytic continuation
 cannot be obtained (c.f.\! \cite[Section 4]{Lechner08}).
 \item Note that $\M$ is full and hence not hyperfinite \cite[Theorem 4.1]{Ueda11-1}.
 Therefore, when we consider the inclusions $\Ad T(a)(\M)\subset \M$, the split property must fail for small $a$.
 \item Even though each $(\M_\k, T_\k, \Omega_\k)$ has the split property,
 namely there are type I factors $\R_{\k,a}$ such that $\Ad T_\k(a)(\M_\k)\subset \R_{\k,a}\subset \M$, one cannot directly obtain an intermediate type I factor from those $\R_{\k,a}$, c.f.\! \cite[Section 4.2]{Tanimoto14-1}, since no type I free product factors can arise (see \cite[Theorem 4.1]{Ueda11-1}).
 
\end{itemize}

\paragraph{Two-particle S-matrix}
Although we do not know whether there is a corresponding Haag-Kastler net,
it is still possible to define the two-particle S-matrix of a Borchers triple.
We follow the approach of \cite[Section 3]{BBS01}\cite[Section 4]{Lechner03}. For simplicity, we take again
two copies of the free field Borchers triple $(\M_\k, T_\k, \Omega_\k), \k = 1,2$.
As the free product von Neumann algebra is generated by the components $\M_\k = \{e^{i\phi_\k(f)}: \supp f \subset W_{\mathrm R}\}''$,
we can construct the asymptotic fields using unbounded fields $\phi_\k$ as follows.

For a test function $f$, its velocity support is defined by $\Gamma(f) := \{(1,p_1/\omega(p))\in \RR^2: p\in\supp \tilde f\}$,
where $\tilde f(p) := \frac1{2\pi}\int d^2a\,e^{ip\cdot a}f(a)$ and $\omega(p) = \sqrt{p_1^2 + m^2}$.
We say $\Gamma(g) \prec \Gamma(f)$ if $\Gamma(f) - \Gamma(g) \subset W_{\mathrm R}$.
If $\Gamma(g)\prec\Gamma(f)$, then it follows that $\supp f^+ - \supp g^+ \subset \RR_+$.
We also introduce a family of time-depending functions
\[
 f_t(a) := \frac1{2\pi}\int d^2p\,\tilde f(p)e^{i(p_0-\omega(p)t - p\cdot a)}.
\]

Let us denote $f^+_\k = \phi_\k(f)\Omega$, by identifying $\H_\k$ as subspaces $\CC\Omega \oplus \H_\k^\circ$ of $\H$
and $\phi'_\k(g) := Z\phi_\k(g)Z$.
These operators can act on $\H$ naturally, by extending the actions
$\l_\k$ and $\rho_{\k'}$ to unbounded operators.
It holds that $\phi_\k(f)\Omega = f^+_\k = f^+_{\k,t} = \phi_\k(f_t)\Omega$
and $\phi'_\k(g)\Omega = g^+_\k = g^+_{\k,t} = \phi_\k(g_t)\Omega$.

Let $f$ and $g$ be test functions as above and we assume further that
their Fourier transforms $\tilde f, \tilde g$ are supported in some neighborhoods of the mass shell
$\{(\omega(p_1), p_1) \in \RR^2\}$.
Now we can define the outgoing two-particle states by
\begin{align*}
 (g_{\k'}^+\times f_\k^+)_\tout &:= \lim_{t\to \infty} \phi'_{\k'}(g_t)\phi_\k(f_t)\Omega, \\
 (f_\k^+\times g_{\k'}^+)_\tout &:= \lim_{t\to \infty} \phi_\k(f_t)\phi'_{\k'}(g_t)\Omega.
\end{align*}
As $\Gamma(g) \prec \Gamma(f)$, $\phi_\k(f)$ and $\phi'_{\k'}(g)$ almost commute as $t \to \infty$
(see \cite[Section 3]{BBS01} for the precise statement) and we get the bosonic statistics
$(g_{\k'}^+\times f_\k^+)_\tout =  (f_\k^+\times g_{\k'}^+)_\tout$.
The incoming two-particle states are similar: by exchanging the roles of $\phi_\bullet$ and $\phi'_\bullet$ and
by noting that
$\phi'_\k(f)$ and $\phi_{\k'}(g)$ almost commute as $t \to -\infty$:
\begin{align*}
 (g_{\k'}^+\times f_\k^+)_\tin &:= \lim_{t\to -\infty} \phi_{\k'}(g_t)\phi'_\k(f_t)\Omega, \\
 (f_\k^+\times g_{\k'}^+)_\tin &:= \lim_{t\to -\infty} \phi'_\k(f_t)\phi_{\k'}(g_t)\Omega,
\end{align*}
and the two-particle S-matrix is defined by \footnote{Actually there are several definitions of S-matrix.
Our S-matrix is a partial isometry on the Hilbert space of the model, while others are defined
on the symmetric Fock space, see e.g.\! \cite[Section6]{Lechner08}.
As the incoming and outgoing two-particle spaces are orthogonal in our models, the latter definition would give
the zero operator.}
\[
 S(g_{\k'}^+\times f_\k^+)_\tout = (g_{\k'}^+\times f_\k^+)_\tin.
\]

It is easy to explicitly compute $S$.
Indeed, when $\k = \k'$, the operators act on $\H_\k$ and $Z$ has no effect, therefore,
the field operators are just the free field and we obtain $(g_{\k}^+\times f_\k^+)_\tout = (g_{\k}^+\times f_\k^+)_\tin$.

The case $\k=1, \k' = 2$ is more interesting.
In the definition of $(g_{\k'}^+\times f_\k^+)_\tout$ above,
the right-hand side is the product of operators $\phi_\k(f), \phi'_{\k'}(g)$,
acting from the left and the right, respectively.
Note also that $\<\Omega_\k, \phi_\k(f)\Omega_\k\> = 0 = \<\Omega_{\k'}, \phi_{\k'}(g)\Omega_{\k'}\>$.
Therefore, we simply get
\[
 (g_{\k'}^+\times f_\k^+)_\tout = \phi'_{\k'}(g)\phi_\k(f)\Omega
 = f^+\otimes g^+ \in \H_{\k,1}\otimes \H_{\k',1} \subset \H_{\k}^\circ\otimes\H_{\k'}^\circ,
\]
where $\H_{\k,1}, \H_{\k',1}$ are the one-particle Hilbert spaces. 
Note that $\supp f^+ - \supp g^+ \subset \RR_+$, and as $\H_{k,1} = \H_{\k',1} = L^2(\RR)$,
hence they can be regarded as a two-variable function $\psi(\theta_1, \theta_2)$
with $\supp \psi \subset \{(\theta_1, \theta_2)\in \RR^2: \theta_1 \ge \theta_2\}$.
It is also immediate that any $L^2$-function with such support can be obtained
as the limit of linear combinations of scattering states.
As for the incoming state, we exchange the roles of $\phi_\bullet$ and $\phi'_\bullet$, and
\[
 (g_{\k'}^+\times f_\k^+)_\tin = \lim_{t\to -\infty} \phi_{\k'}(g_t)\phi'_\k(f_t)\Omega
 = g^+ \otimes f^+ \in \H_{\k',1}\otimes \H_{\k,1},
\]
again with $\supp f^+ - \supp g^+ \subset \RR_+$. Such two-particle scattering states
span the subspace of two-variable functions $\psi(\theta_1,\theta_2),
\supp \psi \subset \{(\theta_1,\theta_2)\in \RR^2: \theta_1 \le \theta_2\}$.
The action of $S$ is then simply the extension of $S(g^+\otimes f^+) = f^+\otimes g^+$.

To summarize,
\begin{align*}
 S_{\k,\k}& = \1 \mbox{ on } \H_{\k,2} \subset \H_\k^\circ,  & \k = 1,2,\\
 S_{\k,\k'}&: \H_{\k,1}\otimes \H_{\k',1} \ni \psi(\theta_1,\theta_2) \longmapsto \psi(\theta_2,\theta_1)
 \in \H_{\k',1}\otimes \H_{\k,1}, \\
 &\qquad \qquad \qquad \mbox{ where }
 \supp \psi(\theta_1,\theta_2)  \subset \{(\theta_1, \theta_2)\in \RR^2: \theta_1 \ge \theta_2\}, &\k \neq \k'.
\end{align*}
Especially, the S-matrix is {\it nontrivial} and {\it not asymptotically complete}, 
i.e.\! not defined on the whole two-particle space $\H_{1,2}\oplus\H_{2,2}\oplus \H_{\k,1}\otimes \H_{\k',1} \oplus \H_{\k',1}\otimes \H_{\k,1}$,
but only on a proper subspace of it.

\subsubsection*{Acknowledgement}
We thank Detlev Buchholz, Gandalf Lechner and Mih\'aly Weiner for stimulating discussions.

Y.T.\! would like to thank the Isaac Newton Institute for Mathematical Sciences for
support and hospitality during the programme ``Operator algebras: subfactors and their applications'',
supported by EPSRC Grant Number EP/K032208/1 and partially by a grant from the Simons Foundation,
where part of this work was undertaken.

\appendix
\section{Absence of conditional expectation between local algebras in M\"obius covariant nets}\label{noexpectation}
Let $(\A, U, \Omega)$ be a local M\"obius covariant net on $S^1$.
It follows that $\A(I)$'s are factors (from the uniqueness of the vacuum).
Let $I_1 \subsetneq I_2$ be an inclusion of intervals on $S^1$. We show that there is
no normal faithful conditional expectation from $\A(I_2)$ onto $\A(I_1)$.

To show this by contradiction, let $E$ be such a conditional expectation. Take a faithful normal state $\psi$ on $\A(I_2)$ which is invariant under $E$
(one can take a faithful normal state on $\A(I_1)$ and extend it to $\A(I_2)$ by $E$).
We may assume that it is a vector state represented by a cyclic and separating vector $\Psi$ thanks to the fact that every representation of a type III factor is standard. 

Now, by Takesaki's theorem \cite[Theorem IX.4.2]{TakesakiII}, $\K := \overline{\A(I_1)\Psi}$ is a proper subspace.
As $\A(I_1)$ is a factor, and since $\A(I_1)\vee (\A(I_1)'\cap \A(I_2)) = \A(I_2)$
(conormality, \cite[Theorem 1.6]{GLW98}), 
$\A(I_1)$ and $\A(I_1)'\cap \A(I_2)$ are in a position of tensor product:
this follows from the fact that $E(\A(I_1)'\cap \A(I_2)) = \CC\1$ by factoriality,
and that $\A(I_1)\vee (\A(I_1)'\cap \A(I_2)) = \A(I_2)$ spans the whole Hilbert space from $\Psi$.
This is impossible, because $\A(I_1)'\cap \A(I_2)$ contains another local algebra $\A(I_3)$ where $I_3$ is an interval
which has a boundary point in common with $I_1$
($I_1$ and $I_2$ can have either one or no point of boundary in common, and either case
one can find $I_3$). On the other hand, it is well known (e.g.\! \cite[P292, footnote]{Buchholz74})
that $\A(I_1)$ and $\A(I_3)$ are not in the position of tensor product (split),
hence we got a contradiction.

Even in higher dimensions, if one assumes the split property and that local algebras are of type III$_1$,
which are quite generic (see e.g.\! \cite{BDF87}), the following argument excludes the possibility that
there is a conditional expectation from $\A(O_2)$ onto $\A(O_1)$ for all local regions $O_1 \subset O_2$.
Indeed, by composing two such expectations, we may assume that $O_1 \Subset O_2$,
and now there is an intermediate type I factor $\R$ such that $\A(O_1) \subset \R \subset \A(O_2)$.
Now the expectation restricts to $\R$, which is impossible by a similar argument as above,
since $\R = \A(O_1) \vee (\R \cap\A(O_1)')$ because $\R$ is a type I factor, while $\A(O_1)$ is of type III$_1$.

{\small

\def\cprime{$'$} \def\polhk#1{\setbox0=\hbox{#1}{\ooalign{\hidewidth
  \lower1.5ex\hbox{`}\hidewidth\crcr\unhbox0}}} \def\cprime{$'$}

}

\end{document}